\documentclass[11pt]{article}
\usepackage[algo2e,ruled,vlined]{algorithm2e}
 \usepackage{color}
\usepackage{amsthm}
\usepackage{mathtools}
\usepackage[caption=false]{subfig}

\newtheorem{theorem}{Theorem}
\newtheorem{lemma}{Lemma}
\newtheorem{problem}{Problem}

\newtheorem{definition}{Definition}

\newtheorem{remark}{Remark}

\newcommand{\1}{\mathbf{1}}

\renewcommand{\emptyset}{\varnothing}

\def\diag#1{\mathrm{diag}\left(#1 \right)}

\def\defeq{\stackrel{\mathrm{def}}{=}}

\def\expec#1#2{{\mathbb{E}}_{#1}\left[ #2 \right]}

\newcommand{\calB}{\mathcal{B}}
\newcommand{\calG}{\mathcal{G}}

\newcommand{\calC}{\mathcal{C}}
\newcommand{\calW}{\mathcal{W}}
\newcommand{\calT}{\mathcal{T}}

\DeclareMathOperator*{\minimize}{minimize}
\DeclareMathOperator*{\maximize}{maximize}

\def\trace#1{\mathrm{Tr} \left(#1 \right)}

\def\sizeof#1{\left|#1  \right|}
\usepackage{fullpage}
\usepackage{mathtools}
\usepackage[margin = 2.5cm]{geometry}

\usepackage{amsmath, amssymb, amsthm}

\usepackage{graphicx}
\usepackage{color}
\usepackage[dvipsnames]{xcolor}

\usepackage{pgfplots}
\pgfplotsset{compat=1.15}

\usepackage{tikz}
\usepackage[hyperindex,breaklinks]{hyperref}

\usepackage{mdframed}

\usepackage[shortlabels]{enumitem}

\usepackage[utf8]{inputenc} 
\usepackage[T1]{fontenc}    

\usepackage{amsmath,amsthm,amsfonts,bbm}

\usepackage[round]{natbib}

\usepackage{pgfplots}
\pgfplotsset{compat=1.15}

\usepackage{tikz}
\usepackage{graphicx}
\usepackage{caption}

\usepackage{algorithm}
\usepackage{algorithmic}

\usepackage{cleveref}
\usepackage{hyperref} 
\usepackage{url}

\usepackage{mathtools}
\usepackage{booktabs}
\usepackage{microtype}

\usepackage{xfrac}
\usepackage{nicefrac}    

\usepackage{dsfont}

\title{Dynamic Curing and Network Design in SIS Epidemic Processes}

\author{Yuhao Yi\footnote{Sichuan University, Chengdu, China,  yuhaoyi@scu.edu.cn. Part of this work was done when Y. Yi was with KTH Royal Institute of Technology, and was also affiliated with Digital Futures.} \and Liren Shan\footnote{Toyota Technological Institute at Chicago, Chicago, IL, USA, lirenshan@ttic.edu} \and
Shijie Wang\footnote{The Chinese University of Hong Kong, Ma Liu Shui, Hong Kong, China,   2019141460148@stu.scu.edu.cn. This work was done when S. Wang was with Sichuan University.}
\and
Philip E. Par\'e\footnote{Purdue University, West Lafayette, IN, USA, philpare@purdue.edu} \and Karl H. Johansson\footnote{KTH Royal Institute of Technology, Digital Futures, Stockholm, Sweden, kallej@kth.se}}

\begin{document}
\maketitle







\begin{abstract}
This paper studies efficient algorithms for dynamic curing policies and the corresponding network design problems to guarantee the fast extinction of epidemic spread in a susceptible-infected-susceptible (SIS) model. We consider a Markov process-based SIS epidemic model. 
We provide a computationally efficient curing algorithm based on the curing policy proposed by Drakopoulos, Ozdaglar, and Tsitsiklis (2014). 
Since the corresponding optimization problem is NP-hard, finding optimal policies is intractable for large graphs. 
We provide approximation guarantees on the curing budget of the proposed dynamic curing algorithm. We also present a curing algorithm fair to demographic groups. 

When the total infection rate is high, the original curing policy includes a waiting period in which no measure is taken to mitigate the spread until the rate slows down. 
To avoid the waiting period, we study network design problems to reduce the total infection rate by deleting edges or reducing the weight of edges. Then the curing processes become continuous since the total infection rate is restricted by network design.
We provide algorithms with provable guarantees for the considered network design problems. In summary, the proposed curing and network design algorithms together provide an effective and computationally efficient approach that mitigates SIS epidemic spread in networks.
\end{abstract}




\section{Introduction}
The modeling and control of contagious processes have received significant research interest due to the overwhelming societal cost of widespread epidemics.
Extensive mathematical models have been proposed to characterize the behavior of human pathogens, computer malware, or misinformation. The susceptible-infected-susceptible (SIS) model~\citep*{kermack1932contributions} and the susceptible-infected-recovered (SIR) model~\citep*{youssef2011individual} are the simplest and most popular models, for which the analysis and control of epidemic spread are investigated accordingly. In an SIS model, individuals can be infected multiple times. Network-based compartmental models have been studied for SIS processes, including stochastic models~\citep*{van2009virus,GMT05} and their mean-field approximations~\citep*{ahn2013global,pare2017epidemic}.

In a wave of the epidemic, the demand for medical services could outrun local resources including staff, supplies, equipment, hospital space, vaccines, and antivirals. Under such circumstances, policymakers sometimes need to coordinate regional resource allocation and usage to improve medical services~\citep*{TRFWP+12,DYSSM+22}. Therefore, optimal resource allocation problems have been studied for various epidemic models. In practice, medical resource allocation strategies are complemented by non-pharmaceutical strategies~\citep*{HBHTA+20} to reduce the number of infections and therefore the demand for medical resources. For example, travel restrictions, mask wearing, mass screening, school closure, and stay-at-home policies. Simulation studies have shown the advantages of combined medical and non-pharmaceutical countermeasures in effectively mitigating epidemic spread~\citep{HBHTA+20}. However, the problem of allocating curing resources and designing the contact network optimally in an SIS model remains to be investigated from a theoretical viewpoint to attain better understanding of their interaction.

In this paper, 
we study a Markov process-based SIS model~\citep*{GMT05} on a weighted undirected graph. Each edge represents the contact between two individuals, and the weight of the edges is proportional to its infection rate. We further investigate algorithms of intervention that are both effective and computationally efficient to guarantee the fast extinction of SIS processes. We provide algorithms for various interventions, including curing policies and contact restrictions. For curing policies, the resource constraint is the curing budget. To achieve fast extinction, an ordering of curing must be computed such that the total infection rate is maintained relatively small~\citep*{DOT14}. Thus, we aim to calculate a near-optimal ordering with an efficient algorithm.
We also consider contact restrictions to mitigate the spread concurrently with curing policies. 
The contact restriction problems can be formulated as network design problems. The restriction of contact between two individuals is modeled as the reduction of the edge weight. The cost of contact restriction is the sum of edge weight reductions. For network design problems with the cost constraint, we present approximation algorithms to the optimal solutions. 
In the end, the algorithms for curing policies and network designs can be integrated to guarantee the fast extinction of the epidemic spread, which shows the superiority of the combined medical and non-pharmaceutical countermeasures.

Both curing and network design algorithms for epidemic interventions have been proposed for various models. While the problems have been studied from various perspectives, most of the proposed policies are static, which means that the policy is fixed given the initial configuration.

For the curing policy, \citet*{PZEJ+14} studied a static resource allocation problem to minimize the cost to bound the spectral radius of a matrix in a mean-field approximation of an SIS model. The cost was defined as a convex function of the infection rate reductions. The authors formulated the problem as a semi-definite program that can be solved in polynomial time. Recently, ~\citet*{JV23} presented an algorithm with lower time and space complexity than the SDP-based algorithm in~\cite{PZEJ+14}. \citet*{WLPCQBB20} proposed the optimal strategy to stabilize SIS processes by allocating curing resources to a single node. 
Spread minimization problems are also studied in SIS models without using mean-field approximation. \citet*{BCGS10} proposed to allocate curing rates according to the degree of each node. Using this policy, logarithmic extinction time is achieved by allocating total curing rate proportional to the number of edges in the network. 
\citet*{CHT09} improved the algorithm by using the personalized PageRank vectors, which exploit the local structure of networks.

All the approaches mentioned above use a static policy. In contrast, a dynamic curing policy can adapt to the observed history of epidemic spread. Bang-bang control and feedback control policies were studied in~\citet*{KKYV18}. A feedback control policy~\citep*{WGIJ21} was proposed to suppress the SIS process. However, optimal curing resource allocation remains to be investigated rigorously. In this paper, we consider the same setting as in~\cite{DOT14}, where the curing budget is limited at any given moment, and a controller optimally distributes resources among all nodes. 
Drakopoulos et\ al.\ proposed a dynamic curing policy for the stochastic SIS model with budget constraint, namely the CURE policy~\citep{DOT14}. The approach seeks to cure individuals in an optimal order. Following this order, the total infection rate is always bounded by the \emph{impedance} of the initial infected set in the network. The impedance is the largest infection rate among any remaining subset of the initial infected set following the optimal order. Thus, the curing budget required by the policy depends on the \emph{cutwidth} of the network, which is the impedance of the whole network. 

For the network design, algorithms have been proposed for various epidemic models. 
For SI models, a $O(\log^2n)$-approximation algorithm is proposed by~\cite{ST04} to find the minimum weight of edges such that a given number of vertices are separated from the initially infected nodes by removing these edges. Constant factor bi-criteria algorithms for this problem are later given by~\citet*{HKPS05,EGAMSTW04}. Two modified problems considering demographic and individual fairness are investigated recently in~\citet*{DSTV22}. 
For the independent cascade susceptible-infected-removed (IC-SIR) model, recent studies by~\citet*{YSPJ20,BDSTV19} show approximations to minimize the expected number of infections by deleting edges, under various assumptions about the underlying graph. 
For the SIS model, \citet*{SAPVK15} provides approximation algorithms for spectral radius reduction problems by edge or node deletion.

Our paper is developed based on the dynamic curing policy in~\cite{DOT14}. The CURE policy provides insights into curing rate allocation. However, it can be improved in various aspects. First of all, as mentioned by~\cite{DOT14}, computing the impedance is {\bf{NP}}-hard. Therefore efficient approximation algorithms are needed. Secondly, fairness of curing policies is considered a major factor in public health decision-making, which may be compromised by an algorithm with minimum resource consumption. 
Recently, \cite{DSTV22} studied fair graph cuts for disaster containment.  The disaster model considered in \cite{DSTV22} is equivalent to a susceptible-infected (SI) epidemic model. In this paper, we study fair curing algorithms in a stochastic SIS model.
Moreover, the impedance of an infected set can be greater than the curing capability. When the impedance is large, the policy sets a waiting period. We address the first two issues by designing computationally efficient curing algorithms under resources and fairness constraints, and the last issue by allowing a controller to modify network connections.

\subsection{Contributions}
The main contribution of this paper is on providing efficient algorithms for curing policies and network design. 
We extend the $O(\log^{2}{n})$ approximation algorithm for cutwidth in \cite{LR99} to the calculation of impedance, which was mentioned in~\cite{DOT14} as an open problem. 
We further propose a variant of the problem that takes into consideration demographic fairness. We introduce the fair impedance and provide $O(\log^2 k)$ approximations for the fair impedance.
The algorithms are constructed by designing a dynamic programming solution based on a tree embedding~\citep{Rac08}.

In addition, we study the problem of dynamically modifying the graph structure to ensure that the given curing budget is adequate to achieve fast extinction and avoid the waiting period. We consider two settings: the SIS model with and without targeted curing rate allocation. For the first case, we propose algorithms to minimize the cut between infected nodes and susceptible nodes for a fixed order of nodes. For the latter case, we propose a $1.44$-approximation algorithm for the problem of minimizing the maximum cut. 
By combining the proposed curing and network design algorithms, we provide a comprehensive solution to control SIS processes in networks.

\subsection{Outline}
The remainder of the paper is organized as follows. 
In Section~\ref{pre.sec} we introduce notations and the considered epidemic model. In Section~\ref{form.sec} we propose the considered policies and corresponding optimization problems. In Section~\ref{algo.sec} we propose algorithms for calculating crusades for the CURE policy and its provably fair counterpart. Section~\ref{design.sec} presents algorithms for the edge weight reduction problems with and without targeted curing policies. Section~\ref{exper.sec} shows the effectiveness of the proposed algorithms by numerical simulations using real network data. Section~\ref{concl.sec} summarizes the presented results, and discusses open problems.
\section{Preliminaries}
\label{pre.sec}
In this section we first introduce some frequently used notations and definitions. Then we describe the considered SIS model.

\subsection{Concepts and Notations}
\label{notations.subsec}
We consider an SIS epidemic process in an undirected weighted graph $G=(V, E, w)$, in which $V$ is the set of nodes and $E$ is the set of edges with a weight function $w:E \mapsto [0,1]$. The degree $d_u$ of a node $u$ is defined as the sum of weights of all its incident edges. We denote by $d_{\max}$ the maximum degree of nodes in the graph. 
We adopt some of the terms used in~\cite{DOT14,DOT17}. We refer to a subset of $V$ as a \emph{bag}. For any subset $A$ of $V$, we use $A^c \defeq  V \backslash A$ to denote the complement of $A$. For any subset $A$ and any node $u \in A^c$, let $A+u \defeq A\cup \{u\}$ and $A-u \defeq A\backslash \{u\}$. We then define the cut of the graph $G$ with respect to a bag. We denote by $G[A]$ the subgraph of $G$ supported on the bag $A$, i.e.\ $G[A]=(A,E', w')$ where $E'=\{(u,v)\mid u,v\in A\}$ and $w'(e)=w(e)$ for $e\in E'$.

\begin{definition}
A cut of the graph $G$ is defined for a bag $A$ as the vertex partition $(A, A^c)$. The size of the cut is defined as
    $c(A)\defeq \sum_{(u,v)}w_{uv}$,
where $(u,v)\in E$, $u \in A$, $v\in A^c$, and $w_{uv}$ is the weight of the edge $(u,v)$.
\end{definition}

We also use the standard definition for the balanced cut in \cite{WS11}, where a cut is $\alpha$-balanced if $\min\{\sizeof{A},\sizeof{A^c}\} \geq \alpha \sizeof{V}$.
We introduce the following definition of maximum restricted cut.
\begin{definition}
Given a graph $\calG$ and a bag $A$, the maximum restricted cut of the bag $A$ is defined by
\begin{align}
    \phi(A)\defeq\max_{Q\subseteq A} c(Q)\,.
\end{align}
\end{definition}

We further recall some concepts defined in~\cite{DOT14,DOT17}.

\begin{definition}
For any two bags $A$ and $B$ satisfying $B\subseteq A$, a monotone crusade from $A$ to $B$ 
is a sequence of bags $p(A,B)\defeq(p_0,p_1,\ldots,p_k)$ where $p_0=A$, $p_k=B$; $p_{i} \subseteq p_{i-1}$ and $|p_{i-1}\backslash p_{i}|=1$ for any $i\in[k]$. We denote by $\calC(A,B)$ the set of all crusades from $A$ to $B$.
\end{definition}

\begin{definition}
The \emph{width} of a crusade $p=(p_0,p_1,\ldots, p_k)$ is defined by $z(p) \defeq \max_{0\leq i\leq k}\{c(p_i)\}$.
\end{definition}

\begin{definition}
Given a bag $A$, its \emph{impedance} $\delta(A)$ is defined as $\delta(A)\defeq \min_{p\in \calC(A,\emptyset)}z(p)$, namely the minimum width of a crusade from $A$ to $\emptyset$. 
\end{definition}


The impedance $\delta(V)$ is referred to as the \emph{cutwidth} of a graph, denoted by $\calW$. The problem of calculating the cutwidth is also referred to as the minimum cut linear arrangement problem~\citep{LR99}. Sometimes we use a superscript to indicate the associated network of a quantity, for example, $z_{G}(p)$, $\delta_{G}(A)$, and $\phi_{G}(A)$.

To discuss fairness of curing policies, we assign vertices in the graph to demographic groups. Suppose we are given sets $V_1$, $V_2$, $\ldots$, $V_\ell$, which composite the vertex set $V$. We assume that the number of groups $\ell$ is a constant. Let $p=(p_0,p_1,\ldots,p_k)$ be a crusade from $A$ to $B$. We consider a set $\calT$ of checkpoints $0 < \tau_1 <\tau_2 <\dots < \tau_s < k$, on which we want to check the fairness of curing policies. 
Let $\tau_0 = 0$ be a dummy checkpoint. For every $i \in \{0,1,2,\cdots,s-1\}$, we use $S_i = p_{\tau_{i}}\backslash p_{\tau_{i+1}}$ to denote the set of nodes between checkpoints $\tau_{i}$ and $\tau_{i+1}$.
For any $\gamma \geq 1$, the crusade $p$ is a $\gamma$-fair crusade with respect to the set of checkpoints $\calT$ 
if for every $i \in \{0,1,2,\cdots,s-1\}$, the set of nodes $S_{i}$ satisfies for every $h\in[\ell]$,
\begin{align}
\label{fairCrus.def}
\sizeof{S_{i} \cap V_h}
< \gamma \cdot  \frac{\sizeof{p_0\cap V_h}}{\sizeof{p_0}} \cdot \sizeof{S_{i}} + 1\,,
\end{align}
which means the curing ratio of a group $h$ in the set between checkpoints $\tau_{i}$ and $\tau_{i+1}$ is preserved within $\gamma$ times its initial ratio. We also call these sets $S_0,S_1,\cdots,S_{s}$ as a $\gamma$-fair partitioning of the bag $A\backslash B$. 

We denote by $\calC_{\gamma,\calT}(A,B)$ the set of all $\gamma$-fair crusade from $A$ to $B$ with respect to checkpoints $\calT$, then we define the notion of \emph{$\gamma$-fair impedance}.
\begin{definition}
Given a bag $A$ and a set of checkpoints $\calT$, its $\gamma$-fair impedance $\delta_{\gamma,\calT}(A)$ is defined as $\delta_{\gamma,\calT}(A)\defeq\min_{p\in \calC_{\gamma,\calT}(A,\emptyset)}z(p)$, the minimum width of a $\gamma$-fair crusade from $A$ to $\emptyset$ w.r.t. $\calT$. 
\end{definition}
We further denote $\1\{\cdot\}$ the indicator function of a given statement, which is equal to $1$ if the statement is true and $0$ otherwise.

\subsection{Epidemic Model}
\label{model.subsec}
We consider a networked SIS model in which each node can be in one of two states: Susceptible or Infected. The infection spreads according to a continuous time Markov process $\{I(t)\}_{t\geq 0}$ on the state space $\{0,1\}^V$ where $I_u(t)=0$ if node $u$ is susceptible and $I_u(t)=1$ if it is infected. Without ambiguity, we also use $I(t)$ to denote the set of all infected nodes at time $t$.
 
The process starts with a given initial state $I(0)$. Let $w_{uv}$ be the infection rate for each edge $(u,v) \in E$. At any time $t\geq 0$, the transition rate of a susceptible node $v$ to the infected state is defined by $\sum_{(u,v)\in E} w_{uv} I_u(t)$, which is the sum of infection rates $w_{uv}$ over all infected neighbors of $v$.

On the other hand, the curing rate of an infected node $u$ is denoted by $\rho_u(t)\geq 0$. We consider two typical scenarios: (a) $\rho_u(t)$ is decided by a network controller with the budget constraint $\sum_{u\in V}\rho_u(t)\leq r$ where $r$ is the curing budget; (b) the allocation of $\rho_u(t)$ for each $u$ is decided by the environment and can be adversarial. We define the total curing capacity be $r(t) \defeq \sum_{u\in V}\rho_u(t)$.  
\section{Problem Formulations}
\label{form.sec}
In this section we present formal definitions for the optimization problems that arise in curing ordering computation and targeted contact restrictions. 
\subsection{Curing Policies}
\label{curing.subsec}
The general strategy of the considered curing policies is to assure that the process has a downward drift most of the time during the process to guarantee short extinction time. Since curing infected nodes may increase the cut between the infected and susceptible nodes, the ordering of the nodes getting cured has a significant impact on the process. The CURE policy proposed in~\cite{DOT14} seeks to find paths along which the cuts are maintained as small as possible such that the spread stops in sublinear time with high probability. We briefly recall the CURE policy as follows:
\begin{itemize}
    \item Wait until $c(I(t))\leq r/8$. Let $A$ be the set of infected nodes right after the waiting period.
    \item Start a segment. Calculate the optimal crusade and obtain an ordering $\{v_1,\dots, v_{|A|}\}$ in the beginning of a segment. Let $C\defeq \{v_2,\dots, v_{|A|}\}$. At any $t$ before the segment ends, allocate the entire curing budget to cure an arbitrary node in the set $D(t)\defeq I(t)\backslash C$. A segment ends when $I(t)=C$ or $|D(t)|\geq r/(8d_{\max})$. If $I(t)=C$, a new segment will be started. If $|D(t)|\geq r/(8d_{max})$, a waiting period will be started.
\end{itemize}
\cite{DOT14} does not specify the algorithm to calculate the optimal crusade.
A dynamic programming algorithm with exponential time and space complexity is obtained by directly applying the Bellman equation given in the paper. \cite{DOT14} listed the problem of improving the computational complexity of the algorithm as an open problem. In this paper we study 
the following optimization problem.


\begin{problem}
Given a network $G$ and a bag $A$, find a crusade $p$ from $A$ to $\emptyset$ such that
\begin{align}
    \minimize_{p\in C(A,\emptyset)}\quad &
    z(p)\,.
\end{align}
\end{problem}
It has been shown that the problem of calculating the cutwidth (and hence the impedance) of a graph is {\bf{NP}}-hard~\citep{Gav77}. Therefore we 
resort
to approximation algorithms which compute a crusade whose width is bounded within a certain factor compared to the minimum width.



By integrating demographic information, a $\gamma$-fair curing policy only considers $\gamma$-fair crusades as feasible target paths of curing. Apart from this distinction, the policy follows a similar 
procedure
as the CURE policy. Thus we also study approximation algorithms for the following problem. 

\begin{problem}
Given a network $G$, a bag $A$, a number $\gamma\geq 1$, and a set of checkpoints $\calT$, find a crusade $p$ from $A$ to $\emptyset$ such that
\begin{align}
    \minimize_{p\in C_{\gamma,\calT}(A,\emptyset)}\quad &
    z(p)\,.
\end{align}
\end{problem}

\subsection{Network Design}
\label{designForm.subsec}
A drawback of the CURE policy is that when $c(I(t))>r/8$, the policy starts a waiting period. This is clearly undesired since the epidemic spreads to more nodes while no measures are taken. For similar reasons, the CURE policy only works when the budget is $\Omega(\calW)$, where $\calW$ is the cutwidth of the network. Moreover, the authors have shown in subsequent work~\citep*{DOT15,DOT17} that there exists a phase transition such that when $r=o(\calW)$ the extinction time is exponential regardless of the curing policy.

For public health decision-making, it is advisable to consider both medical and non-pharmaceutical intervention strategies. Real-world epidemic response is often a composition of medical resource allocation and contact restrictions. Previous work has shown that, by applying along with preferentially vaccinating urban locations,
travel restrictions can postpone the arrival of peak in the spread of influenza~\citep*{EEFF19}. School closure and social distancing are combined with antiviral treatment or vaccination to reduce the cost of an influenza pandemic~\citep*{KHM13}. A simulation of smallpox shows the effectiveness of combining vaccination with school closure~\citep*{ZI10}. Motivated by these simulation studies, we investigate the problem of reducing network connections to achieve fast extinction with insufficient medical resources.

Given a graph $G$ and a budget $r$, if $r$ is insufficient for any policy to suppress the spread, we ask the question of how to modify the network to effectively reduce its impedance and cutwidth. Since modifying the network results in the change of optimal curing order, and changing the curing order affects the optimal weight reduction of the network, optimizing the two policies simultaneously is challenging. We consider a simplified problem with any fixed curing order\footnote{In practice we use the ordering attained by applying curing policies to the original contact network.}.

Given the CURE policy, we study an alternative problem, in which we modify the graph to effectively reduce the width of a \emph{fixed} crusade computed in the beginning of each curing segment. Specifically, we modify the graph by reducing the weight of each edge $(u,v) \in E$ to $w'_{uv} = w_{uv}-\Delta_{uv}$. Our goal is to minimize the total weight reduction such that the width of the current crusade is no more than a threshold $b$.  

\begin{problem}
\label{fractal.prob}
Given a graph $G$ with the edge weight function $w$, a bag $A$ and a crusade $p$ from $A$ to $\emptyset$, a threshold $b\in \mathbb{R}_+$, find the weight reduction $\Delta_{uv}$ of each edge $(u,v)\in E$ for the following optimization program: 
\begin{align}
    \minimize_{\Delta} \quad & \sum_{(u,v)\in E}\Delta_{uv} ,\\
    \textrm{subject to}\quad  & 0\leq \Delta_{uv} \leq w_{uv}, \forall (u,v)\in E\,, \nonumber\\
    & z_{G'}(p) \leq b\,,\nonumber
\end{align}
where $G'=(V, E, w')$ is the modified graph with weight $w'_{uv}=w_{uv}-\Delta_{uv}$ for $(u,v)\in E$. 
\end{problem}

We also consider this problem under the constraint that weight reductions $\Delta_{uv} \in \{0,w_{uv}\}$ for all edges, which means each edge is either kept or removed completely.  

\begin{problem}\label{integer.prob}
Given a graph $G$ with the edge weight function $w$, a bag $A$ and a crusade $p$ from $A$ to $\emptyset$, a threshold $b\in \mathbb{R}_+$, find the weight reduction $\Delta_{uv}$ of each edge $(u,v)\in E$ for the following optimization program: 
\begin{align}
    \minimize_{\Delta} \quad & \sum_{(u,v)\in E}\Delta_{uv} ,\\
    \text{subject to}\quad  &  \Delta_{uv} \in \{0,w_{uv}\}, \forall (u,v)\in E\,, \nonumber\\
    & z_{G'}(p) \leq b\,,\nonumber
\end{align}
where $G'=(V, E, w')$ is the modified graph with weight $w'_{uv}=w_{uv}-\Delta_{uv}$ for $(u,v)\in E$.
\end{problem}




When the curing policy is decided by the environment, the ordering of curing is arbitrary. Therefore, a variant of the maximum cut of the network decides the extinction time of the SIS process. We also investigate the problem of minimizing the maximum cut of a given bag.

\begin{problem}
\label{mc_fractal.prob}
Given a graph $G$ with the edge weight function $w$, a bag $A$, and a threshold $b\in \mathbb{R}_+$, find the weight reduction $\Delta_{uv}$ of each edge $(u,v)\in E$ for the following optimization program: 
\begin{align}
    \minimize_{\Delta} \quad & \sum_{(u,v)\in E}\Delta_{uv} 
    ,\\
    \text{subject to}\quad  & 0 \leq \Delta_{uv} \leq w_{uv}, \forall (u,v)\in E\,, \nonumber\\
    & \phi_{G'}( A)\leq b\,.\nonumber
\end{align}
where $G'=(V, E, w')$ is the modified graph with weight $w'_{uv}=w_{uv}-\Delta_{uv}$ for $(u,v)\in E$.
\end{problem}

\section{Curing Algorithms}
\label{algo.sec}
In this section we provide algorithms for approximating impedance and $\gamma$-fair impedance of a bag. The definition of these problems are given in Section~\ref{curing.subsec}.

\subsection{Algorithm for CURE Policy}
In this section, we provide a polynomial-time curing algorithm in the network SIS model. Specifically, we combine a polynomial-time approximation algorithm for computing the crusade with the CURE policy proposed in~\cite{DOT14}. We 
present
the following theorem. 

\begin{theorem}\label{thm:extinctiontime}
Given a graph $G$, suppose the curing budget $r \geq \max\{\alpha \calW \log^2 n, 8d_{max} \log n\}$, where $\calW$ is the cutwidth of graph $G$ and $\alpha$ is a fixed constant. Then, there exists a polynomial-time curing algorithm such that the expected extinction time is at most $O(n\log^2 n / r)$.
\end{theorem}

To prove this theorem, we use a polynomial-time algorithm that, given any bag $A$, finds a crusade of $A$ with width at most $O(\log^2 k)$ times the impedance of $A$. This algorithm is shown in Algorithm~\ref{alg:apprImpe}. Our algorithm follows the approaches of approximation algorithms for multiway cut~\citep{Gar96} and cutwidth~\citep{LR99}. The crux of the algorithm is an approximation algorithm for the balanced cut problem~\citep{LR99}. The algorithm recursively calculates a balanced cut on the subgraph supported on a bag. In Algorithm~\ref{alg:apprImpe}, the subroutine BalancedCut returns a partition $(V_1,V_2)$ of the bag $A$, i.e.\ $V_1\cup V_2=A$ and $V_1\cap V_2=\emptyset$. 
Lemma~\ref{lem:balancedcut} shows a property of the partition. The algorithm recursively deals with the subgraphs $G[V_1]$ and $G[V_2]$ until the subgraphs cannot be further divided. Then the algorithm adds elements in $V_2$ to all preceding bags and returns a crusade from $A$ to $\emptyset$.



\begin{lemma}[Section 3.3 in \cite{LR99}]\label{lem:balancedcut}
Given a graph $G$, there exists a polynomial-time algorithm that finds a $1/3$-balanced cut with size at most $O(\log k) \mathrm{OPT}_G(1/2)$, where $\mathrm{OPT}_G(1/2)$ is the size of the minimum balanced cut.
\end{lemma}

\begin{theorem}\label{thm:impedance}
Algorithm~\ref{alg:apprImpe} takes the input of a graph $G$ and an infected bag $A$ with $|A|=k$ and outputs a crusade $p\in \calC(A,\emptyset)$ with width $z(p)$, such that
\begin{align}
    \delta(A) \leq z(p) \leq O(\log^2 k) \delta(A)\,.
\end{align}
\end{theorem}

\begin{proof}
By the definition of impedance, we have the width $z(p) \geq \delta(A)$. For the other side, we only need to upper bound the size of the cut given by the bag $p_i$ for every $i \in \{0,1,2,\cdots, k\}$. Since the bag $p_0 = A$, we have $c(p_0) = c(A) \leq \delta(A)$. 
We use $v_i$ to denote the node in $p_{i-1} 
\setminus p_{i}$ for each $i \in \{1,2,\cdots,k\}$. Each edge in the cut given by $p_i$ is contained in one of the following two cases: (a) the cut $(A,A^c)$; (b) a balanced cut generated by our algorithm over a bag containing the node $v_i$. Consider any balanced cut over a bag $B$ containing $v_i$ in our algorithm. Let $(B_1,B_2)$ be the $1/3$-balanced cut over the bag $B$ returned by the call of BalancedCut$(G[B],B)$. By Lemma~\ref{lem:balancedcut}, we have $$
c_{G[B]}(B_1) \leq O(\log k)\mathrm{OPT}_{G[B]}(1/2).
$$
Suppose $q=(q_0,q_1,\cdots,q_{|A|})$ is an optimal crusade from the bag $A$ to $\emptyset$ with width $z(q) = \delta(A)$. We find the bag $q_j$ in this crusade such that $q_j \cap B$ provides a $1/2$-balanced cut over the set $B$. Thus, we have 
$$
\mathrm{OPT}_{G[B]}(1/2) \leq c_{G[B]}(q_j\cap B) \leq c_G(q_j) \leq \delta(A).
$$
Therefore, every balanced cut in our algorithm has the size at most $O(\log k) \delta(A)$. Then, we show that the total number of recursive balanced cuts containing the node $v_i$ is at most $O(\log k)$. Our recursive algorithm starts from the bag $A$ with $|A|= k$. Since the subroutine BalancedCut returns a $1/3$-balanced cut, the size of the bag containing the node $v_i$ shrinks by a factor of $2/3$ after each recursive call of BalancedCut. Thus, the total number of recursive calls containing the node $v_i$ is at most $O(\log k)$. For every $i \in \{1,2,\cdots,k\}$, the size of the cut given by the bag $p_i$ is at most
$$
c(p_i) \leq c(A)+\sum_{B: v_i \in B} c_{G[B]}(B_1) \leq O(\log^2 k)\delta(A),
$$
which completes the proof.
\end{proof}


We note
that a $O(\log^{2} n)$ approximation algorithm for the cutwidth of a graph is also shown in~\cite{bornstein2004flow} by using a $O(\log n)$ approximation for the minimum linear arrangement (MLA). By using an improved $O(\sqrt{\log n}\log\log n)$ approximation for MLA in~\cite{feige2007improved}, an $O(\log^{3/2} n\log\log n)$ approximation for the cutwidth can be given, which can also be generalized to a $O(\log^{3/2} k\log\log k)$ approximation for the impedance problem. 
To prove Theorem~\ref{thm:extinctiontime}, we show that the following lemma holds by choosing the crusade given by Algorithm~\ref{alg:apprImpe}.

\begin{lemma}
\label{segmentCut.lemma}
For every time $t$ during a segment, the cut given by the set of infected nodes is at most $c(I(t)) \leq r/2$.
\end{lemma}
\begin{proof}
For each segment, we want to reach the target bag $C$ on the current crusade provided by Algorithm~\ref{alg:apprImpe} with respect to an input bag $B$.  Let $D(t) = I(t)\setminus 
C$. Recall that the curing policy will start a waiting period if $|D_t| \geq r/(8d_{max})$. Then, we have
$$
c(I(t)) \leq c(C) + c(D(t)) \leq c(C) + d_{max} \cdot |D(t)| \leq c(C) +r/8,
$$
where the last inequality holds since $|D_t| < r/(8d_{max})$. According to Lemma 1 in~\cite{DOT14}, the impedance of the bag $B$ is upper bounded by $\calW+c(B)$. By Theorem~\ref{thm:impedance}, we have
$$
c(C) \leq \alpha/4 \cdot \log_2^2 n \cdot \delta(B) \leq \alpha/4 \cdot \log_2^2 n(\calW+c(B)).
$$
Since the bag $B$ satisfies our new ending condition for the waiting period, we have $c(B) \leq r/(2\alpha\log_2^2 n)$. By our assumption $r\geq \alpha\calW\log_2^2 n$, we have 
$c(C) \leq 3r/8$, which implies $c(I(t)) \leq r/2$.
\end{proof}

Then, we prove Theorem~\ref{thm:extinctiontime} by modifying the analysis of the CURE policy in~\cite{DOT14}.
\begin{proof}[Proof of Theorem~\ref{thm:extinctiontime}]
We slightly modify the CURE policy as follows. Let the waiting period end when $c(I(t)) \leq r/(2\alpha\log_2^2 n)$, where $\alpha/4$ is the fixed constant in the approximation factor in Theorem~\ref{thm:impedance}. In a segment, we choose the crusade given by Algorithm~\ref{alg:apprImpe} instead of the optimal crusade.  We assume that the curing budget $r \geq \alpha\calW\log_2^2 n$ and $r \geq 8d_{max} \log_2 n$. 

Then, we upper bound the expected extinction time. Since the ending condition for the waiting period is $c(I(t)) \leq r/(2\alpha\log_2^2 n)$, the expected length of a waiting period is at most $(2\alpha n \log_2^2 n)/r$. 
By Lemma~\ref{segmentCut.lemma}, the expected length of a segment is at most $2/r$. The number of segments is at most $n$.
Similar to the analysis in~\cite{DOT14}, we get that the expected extinction time is $O(n \log^2 n /r)$.
\end{proof}

\begin{algorithm2e}[t]
\SetAlgoLined
\caption{ApprImpe$(G,A)$}
\label{alg:apprImpe}
\SetKwInOut{Input}{Input}\SetKwInOut{Output}{Output}
\SetKw{KwBy}{by}
\Input{a graph $G=(V,E,w)$, a bag $A$ with $|A|=k$;}
\Output{a crusade $p\in \calC(A,\emptyset)$;}
\eIf{$|A|=1$}{
\Return $p \gets (\{u\},\emptyset)$;
}{
$(V_1,V_2)\gets \text{BalancedCut}(G[A],A)$;\\
$(q_0, \dots ,q_{|V_1|}) \gets \text{ApprImpe}(G[V_1], V_1)$;\\
$(p_{|V_1|},\dots, p_{|A|}) \gets \text{ApprImpe}(G[V_2], V_2)$;\\
\For{$i \gets 0 $ \KwTo $|V_1|-1$}{
$p_{i} \gets  q_{i}\cup V_2$;
}
\Return $p \gets (p_0, \dots, p_{|A|})$;
}
\end{algorithm2e}

\subsection{Fair Curing Algorithms}

Next, we show an approximation algorithm for finding a fair curing policy. Given a graph $G$, a bag $A$, a fair parameter $\gamma$, and a set of checkpoints $\calT$, our algorithm sequentially finds fair partitions at these checkpoints. Specifically, we start from the first checkpoint $\tau_1$ and find the $\gamma$-fair partition of bag $A$ at checkpoint $\tau_1$. Then, we pick a crusade for nodes before $\tau_1$ with Algorithm~\ref{alg:apprImpe}. After that, we repeat this process on the remaining nodes after $\tau_1$ until there is no checkpoint. Our algorithm is shown in Algorithm~\ref{alg:fairapprImpe}.

\begin{algorithm2e}[t]
\SetAlgoLined
\caption{FairApprImpe$(G,A,\gamma,\calT)$}
\label{alg:fairapprImpe}
\SetKwInOut{Input}{Input}\SetKwInOut{Output}{Output}
\SetKw{KwBy}{by}
\Input{a graph $G=(V,E,w)$, a bag $A$ with $|A|=k$, a parameter $\gamma$, a set of checkpoints $\calT$;}
\Output{a fair crusade $p\in \calC(A,\emptyset)$;}
\eIf{$\calT$ is not empty}{
Find the nearest checkpoint $\tau_i$;\\
Find a $\gamma$-fair partitioning of $A$, $S_{i-1}$  and $A\backslash S_{i-1}$ of sizes $\tau_i$ and $k-\tau_i$ by using dynamic programming on HST;\\

$(q_0, \dots ,q_{|S_{i-1}|}) \gets \text{ApprImpe}(G[S_{i-1}], S_{i-1})$;\\
$(p_{|S_{i-1}|},\dots, p_{|A|}) \gets \text{FairApprImpe}(G[A\backslash S_{i-1}], A\backslash S_{i-1}, \gamma, \calT \backslash \{\tau_i\})$;\\
\For{$i \gets 0 $ \KwTo $|S_{i-1}|-1$}{
$p_{i} \gets  q_{i}\cup (A\backslash S_{i-1})$;
}
\Return $p \gets (p_0, \dots, p_{|A|})$;
}{
$(p_0, \dots, p_{|A|}) \gets \text{ApprImpe}(G[A], A)$;\\
\Return $p \gets (p_0, \dots, p_{|A|})$;
}
\end{algorithm2e}

We first show that our algorithm achieves a $O(\log^2 k)$ approximation for $\gamma$-fair impedance when there is only one checkpoint. 

\begin{theorem}\label{thm:fair}
Given a bag $A$ with size $k$, we have a $O(\log^2{k})$-approximation algorithm for $\gamma$-fair impedance of a bag $A$ in a graph $G$ with respect to groups $V_1,\ldots, V_\ell$ and one checkpoint $\calT = \{\tau_1\}$.
\end{theorem}

For multiple checkpoints, $s > 1$, we reconstruct the crusade after each checkpoint by using our algorithm. We show that when checkpoints satisfy a doubling condition specified in the following theorem, our algorithm still achieves a bi-criteria approximation to the $\gamma$-fair impedance.

\begin{theorem}\label{thm:fair-multiple-checkpoint}
Given a bag $A$ with size $k$, suppose the set of checkpoints $\calT = \{\tau_1,\tau_2,\dots, \tau_s\}$ satisfy that $\tau_i -\tau_{i-1} \geq \tau_{i-1}$ for every $i =2,3,\dots, s$. We have an algorithm that finds a $2\gamma$-fair crusade of bag $A$ that satisfies: 
the
width of this crusade is at most $O(\log^2{k})$ times the $\gamma$-fair impedance of a bag $A$.
\end{theorem}

Analysis for Theorem~\ref{thm:fair} and Theorem~\ref{thm:fair-multiple-checkpoint} are shown in Appendix~\ref{appdx:fair}.

\section{Network Design}
\label{design.sec}
In this section we present algorithms for network design problems posed in Section~\ref{designForm.subsec}.

\subsection{Network Design for Curing Policies}
In this section we consider the problems of designing the structure of the network to minimize the width of a fixed crusade, namely Problems~\ref{fractal.prob} and~\ref{integer.prob}.

For a fixed crusade, Problem~\ref{fractal.prob} is a linear program with a cost function that is linear in the edge weights, and $k$ linear constraints for each cut to be no more than the curing budget. Therefore we obtain the optimal solution in polynomial time. 
\begin{theorem}
There exists a polynomial time algorithm which finds the optimal solution for Problem~\ref{fractal.prob}. 
\end{theorem}

\begin{algorithm2e}[t]
\SetAlgoLined
\caption{WidthOpt($G$, $A$, $p$)}
\label{alg:widthOpt}
\SetKwInOut{Input}{Input}\SetKwInOut{Output}{Output}
\Input{a graph $G=(V,E,w)$, a bag $A$ with $|A| = k$, a crusade $p$ from $A$ to $\emptyset$}
\Output{$\Delta_{uv}\in\{0, w_{uv}\}$ for $(u,v)\in E$}
Let $\{((u,v),\Delta_{uv}^*): (u,v)\in E\}$ be an optimal solution for the LP in Problem~\ref{fractal.prob};\\
$v_i\gets p_{i-1}
\setminus
p_i$ for $i\in[k]$;\\
Set an arbitrary ordering $\{v_{k+1},v_{k+2},\cdots,v_{n}\}$ for nodes in $A^c$;\\
\For{$i\gets 1$ \KwTo $k$}{
Let $\vec{E}_i$ be the sorted list of $\{(v_i,v_j): j>i\}$ in non-increasing order of $j$;\\
$x \gets 0$;\\
\While{$x < \sum_{(u,v)\in \vec{E}_i} \Delta^*_{uv}$}{
Let the edge $(u,v)$ be the first edge in $\vec{E}_i$;\\
$x \gets x + w_{uv}$;\\
$\Delta_{uv} \gets w_{uv}$ and remove $(u,v)$ from $\vec{E}_i$;\\
}
$\Delta_{uv} \gets 0$ for the remaining edges $(u,v)$ in $\vec{E}_i$;
}
\end{algorithm2e}

Next we consider Problem~\ref{integer.prob}, in which $\Delta_{uv}$ takes the value of either $w_{uv}$ or $0$. We provide an algorithm with the following guarantee:

\begin{theorem}
Given a graph $G$, a bag $A$ with $|A|=k$, and a crusade $p$ from $A$ to $\emptyset$, Algorithm~\ref{alg:widthOpt} provides a solution of Problem~\ref{integer.prob} with the total weight reduction $\mathrm{OPT}(G,A,p)+k$, where $\mathrm{OPT}(G,A,p)$ is the optimal solution of Problem~\ref{integer.prob}. 
\end{theorem}

\begin{proof}
We first show that the total weight reduction returned by Algorithm~\ref{alg:widthOpt} is at most $\mathrm{OPT}(G,A,p)+k$. Let $\Delta_{uv}^*$ be the optimal solution for the LP used in Algorithm~\ref{alg:widthOpt}. Since the edge weight is at most $1$, we have for each $i \in \{1,2,\cdots, k\}$, 
$$
\sum_{(u,v)\in \vec{E}_i} \Delta_{uv} \leq \sum_{(u,v)\in \vec{E}_i} \Delta^*_{uv} + 1.
$$
Since Problem~\ref{fractal.prob} is a relaxation of Problem~\ref{integer.prob}, the total weight reduction in the optimal LP solution is at most $\mathrm{OPT}(G,A,p)$. Thus, we have
$$
\sum_{(u,v)\in E} \Delta_{uv} \leq \sum_{(u,v)\in E} \Delta^*_{uv} + k \leq \mathrm{OPT}(G,A,p)+k.
$$

Then, we show that the solution returned by Algorithm~\ref{alg:widthOpt} is a feasible solution for Problem~\ref{integer.prob}. We say that an edge $(u,v)$ covers $v_i$ if exactly one of $u$ and $v$ is in $p_i$. For each $i \in \{1,2,\cdots, k\}$, let $E_i$ denote all edges covering node $v_i$. By the rounding process, we have for any $i \in \{1,2,\cdots, k\}$,
\begin{align}\label{eq:rounding}
\sum_{(u,v)\in \vec{E}_i} \Delta_{uv} \geq \sum_{(u,v) \in \vec{E}_i} \Delta^*_{uv}.
\end{align}
Note that the edges in $E_i \cap \vec{E}_j$ are the first $|E_i \cap \vec{E}_j|$ edges in the sorted list $\vec{E}_j$ for any $j \leq i \leq k$. 
We claim that for any $j \leq i \leq k$,
\begin{align}\label{eq:feasible}
\sum_{(u,v)\in E_i \cap \vec{E}_j} \Delta_{uv} \geq \sum_{(u,v) \in E_i \cap \vec{E}_j} \Delta^*_{uv}.
\end{align}
We prove this claim by considering the following two cases. If all edges $(u,v)$ in $E_i \cap \vec{E}_j$ are rounded to $w_{uv}$, then the inequality~(\ref{eq:feasible}) holds. If there is an edge $(u,v)$ in $E_i \cap \vec{E}_j$ 
that is 
rounded to $0$, then all edges in $\vec{E}_j 
\setminus E_i$ are rounded to $0$ according to the rounding process. So, we get
$$
\sum_{(u,v)\in E_i \cap \vec{E}_j} \Delta_{uv}  = \sum_{(u,v)\in \vec{E}_j} \Delta_{uv} \geq \sum_{(u,v) \in \vec{E}_j} \Delta^*_{uv} \geq \sum_{(u,v) \in E_i \cap \vec{E}_j} \Delta^*_{uv},
$$
where the first inequality is due to the inequality~(\ref{eq:rounding}).

Therefore, we have for the edge set $E_i$ covering node $i$,  $i \in \{1,2,\cdots, k\}$, 
$$
\sum_{(u,v)\in E_i} \Delta_{uv} \geq \sum_{(u,v) \in E_i} \Delta^*_{uv},
$$
which implies that $\Delta_{uv}$ satisfies the width constraint.
\end{proof}

\begin{remark}
Let $\mathrm{LP}^*=\sum_{(u,v)\in E}\Delta^*_{uv}$ where $\Delta^*_{uv}$ 
is 
an optimal solution for Problem~\ref{fractal.prob}. We provide an example that shows the existence of an $O(\mathrm{LP}^*+k)$ additive integrality gap for any algorithm of Problem~\ref{integer.prob} by rounding the solution of Problem~\ref{fractal.prob}. We let $G$ be a path with all edges weighted $1$ and nodes are sorted from one end to the other. Let $p$ be the crusade where $p_i$ consists of the last $k-i$ nodes. Let $b=0.9$, then the LP returns a solution where $\Delta^*_{uv}=0.1$ for all edges. The only feasible integral solution, however, is $\Delta_{uv}=1$ for all edges. Therefore, the additive integrality gap is equal to $9\cdot \mathrm{LP}^*$ (and also $0.9k$).
\end{remark}

We then consider an unweighted  version of Problem~\ref{integer.prob} where $w_{uv}=1$ and $\Delta_{uv}\in\{0,1\}$ for all $(u,v)\in E$, which we refer to as the \underline{U}nweighted Crusade \underline{W}idth \underline{C}ost \underline{M}inimization \underline{P}roblem (UWCMP). We show an algorithm which finds the optimal solution of UWCMP in polynomial time. Details are shown in Appendix~\ref{ued.appdx}.

\begin{remark}
By applying the proposed network design algorithms, we can discard the waiting period in the CURE policy. Instead we modify the network such that $z_{G'}(p)\leq r/4$ and start a segment. The target path is given by $p$. Then we start a new segment. When $|D(t)|\geq z_{G'}(p)\leq r/(4 d_{\max})$, we calculate a new nearly optimal crusade $p$ and a modified graph $G'$ with $z_{G'}(p)\leq r/4$ and start a new segment.
\end{remark}

\subsection{Network Design for Minimizing the Maximum Restricted Cut}
In this section, we consider the problem of designing the structure of the network to minimize the maximum restricted cut of a given bag, namely Problem~\ref{mc_fractal.prob}. Given an algorithm for Problem~\ref{mc_fractal.prob}, we propose a policy to guarantee sublinear expected extinction time, with arbitrary ordering of curing. Let $r(t)$ be the sum of curing rates of all nodes in $I(t)$. Suppose $r(t)$ is always greater than a fixed constant $r'$ for all $t$, the policy is given by:
\begin{itemize}
    \item Network design: modify the graph such that $\phi_{G'}(I(t)) \leq r'/4$ and start a segment. Let $A$ be the set of infected nodes.
    \item Segment: Define $D(t) = I(t) 
    \setminus A$. If $|D(t)|\geq r'/(4d_{\max})-1$ then start a new network design period.
\end{itemize}

\begin{theorem}
\label{thm:mrcut}
If $r' \geq 2\log n$ always holds, by modifying $G'$ such that $\phi_{G'}(I(t)) \leq r'/4$, the policy achieves sublinear expected extinction time.
\end{theorem}

\begin{proof}
We let the curing process start at $t=0$. Let $M_t \defeq  \sizeof{I(t)}$. We define the process $P_t$ on the space $\{0,1,2,\cdots\}$ as follows: Let $P_0 = M_0$. The process $P_t$ has a downward transition rate of $r$ and a rate $r/2$ of upward transition. Let $T$ be the first time that $P_t$ takes the value $0$. We note that $M_t$ is stochastically dominated by $P_t$.

We further let $H_t\defeq M_t+\frac{r}{2}\cdot t$, and $\hat{H}_t$ be the truncation of $H_t$ which stops at time $T$ and keeps the value $H_T$ afterwards. Since the cut between the infected nodes and the healthy nodes is less than $r/2$, and the downward drift is at least $r$, the total downward drift of $\hat{H}_t$ is greater than $r/2$. Therefore $\hat{H}_t$ is a supermartingale.

By Doob’s optional stopping theorem,
\begin{align*}
    k= \expec{}{M_0} = \expec{}{\hat{H}_0} \geq \expec{}{\hat{H}_T} + \frac{r}{2}\expec{}{T} \geq \frac{r}{2} \expec{}{T}\,.
\end{align*}
Therefore,
\begin{align*}
    \expec{}{T} \leq \frac{2k}{r}\,.
\end{align*}
Since $r\geq 2\log n$, we attain  $\expec{}{T}\leq \frac{k}{\log n}\leq \frac{n}{\log n}$.
\end{proof}

Now we present the algorithm for Problem~\ref{mc_fractal.prob}.
We consider the following minimax program for the modified graph $G'=(V,E, w')$ where $w'_{uv}=w_{uv}-\Delta_{uv}$:
\begin{align}
\label{minimax.eqn}
\minimize_{G'} \quad & \phi_{G'}(A)\,,\\
\text{subject to} \quad & 0 \leq \Delta_{uv} \leq w_{uv}, \forall(u,v)\in E\,,\nonumber\\
& \sum_{(u,v)\in E} \Delta_{uv} \leq b'\,,\nonumber
\end{align}
where $b'$ is a given edge reduction budget. We note that by running a binary search on $b'$ we solve Problem~\ref{mc_fractal.prob} with any precision.

\begin{theorem}\label{thm:maxcut_modify}
There exists a polynomial time algorithm which computes a graph $G'$ with weight reductions $\Delta_{uv}$, in which $\phi_{G'}(A)\leq 1.14\cdot \phi_{\tilde{G}}(A)$ where $\tilde{G}$ is any modified graph with the same reduction budget $b$.
\end{theorem}

The maximum restricted cut $\phi_{G}(A)$ for a fixed $G$ and $A$ is given by the following integer program:
\begin{align}
    \maximize_{y_i} \quad& \frac{1}{2}\left(\sum_{\substack{i<j\\i,j\in A}}w_{ij}(1-y_i y_j) + \sum_{\substack{i\in A\\j\in (V 
    \setminus A)}}w_{ij}(1-y_i)\right)\\
    \text{subject to} \quad & |y_i|=1\,,\nonumber
\end{align}
where $y_i\in\{-1,1\}$ for $i\in A$ provides the partition of nodes in the maximum restricted cut.  

Equivalently, we can combine all the nodes in $V 
\setminus A$ as a new node $s$ and obtain a modified graph $\hat{G}$ with $|A|+1$ nodes $\hat{V} = A + s$. The edge set $\hat{E}$ in $\hat{G}$ contains all the edges $(u,v)$ in the graph $G[A]$ with their original weights. We add edge $(u,s)$ to $\hat{E}$ with the weight $\hat{w}_{us}=\sum_{(u,v)\in E, v\in A^c}w_{uv}$ if there exists an edge $(u,v)$ for $v\in A^c$.  Then we study the equivalent  maximum cut problem
\begin{align}
     \maximize_{y_i} \quad& \frac{1}{2}\sum_{\substack{i<j\\i,j\in \hat{V}}}\hat{w}_{ij}(1-y_i y_j)\\
    \text{subject to} \quad & |y_i|=1\,.\nonumber   
\end{align}
Let $L_{\hat{G}}$ be the Laplacian matrix of the graph $\hat{G}$. The above problem has a well-known Semidefinite Programming (SDP) relaxation~\citep{GW95}:
\begin{align}
\label{primal.eqn}
    \maximize_{X} \quad & \frac{1}{4}\trace{L_{\hat{G}}X}\\
    \text{subject to} \quad & X\succcurlyeq 0\,.\nonumber
\end{align}

Therefore, a relaxation of the program in~(\ref{minimax.eqn}) is given by
\begin{align}
    \min_{\Delta} ~~\max_{X} \quad &\trace{B^\top (W-\diag{\Delta}) B X}\label{eq:relax}\\
    \text{subject to} \quad & X\succcurlyeq 0\,, W \succcurlyeq \diag{\Delta} \succcurlyeq 0\,, \nonumber\\
    & \mathbf{1}^\top \Delta \leq b\,,\nonumber
\end{align}
where $B$ is the edge-node incident matrix of $\hat{G}$ and $W$ is the diagonal weight matrix for edges in $\hat{G}$.

\begin{proof}[Proof of Theorem~\ref{thm:maxcut_modify}]
For any maximum cut instance, the optimal value of the SDP relaxation is upper bounded by $1.14$ times the maximum cut~\citep{GW95}. Thus, we have that the optimal value of program~(\ref{eq:relax}) is at most $1.14$ times the optimal value of program~(\ref{minimax.eqn}). Let $(\Delta^*,X^*)$ be the optimal solution of program~(\ref{eq:relax}). 
Let the graph $G'$ be the modified graph corresponding to the weight reduction $\Delta^*$. 
By the minimax theorem~\citep{Neu28,Sio58}, we have 
$$
\phi_{G'}(A) \leq \mathrm{OPT}(\Delta^*,X^*),
$$
where $\mathrm{OPT}(\Delta^*,X^*)$ is the optimal value of program~(\ref{eq:relax}). Thus, we find a modified graph $G'$ such that $\phi_{G'}(A)$ is at most $1.14$ times the optimal solution of program~(\ref{minimax.eqn}).

It suffices to find an algorithm that solves~(\ref{eq:relax}) in polynomial time.
The dual form of~(\ref{primal.eqn}) is given by
\begin{align}
\label{dual.eqn}
    \minimize_{u,t} \quad & \frac{1}{4}n \cdot t\\
    \text{subject to}\quad & tI - (L+\diag{u})\succcurlyeq 0\,,\nonumber\\
    & \mathbf{1}^{\top}u = 0\,.\nonumber
\end{align}
The primal~(\ref{primal.eqn}) and dual~(\ref{dual.eqn}) both have feasible solutions. In addition the primal has a strictly feasible solution. Therefore, weak and strong duality hold for the pair, indicating that the primal and the dual have the same optimal solution.

The program in~(\ref{eq:relax}) can be further written as 
\begin{align}
\label{polySovable.eqn}
    \minimize_{\Delta, u,t} \quad & \frac{1}{4}n\cdot t\\
    \text{subject to} \quad & tI +B^\top \diag{\Delta} B-\diag{u}\succcurlyeq B^T W B\,,\nonumber\\
    & \mathbf{1}^{\top}u = 0\,,\nonumber\\
    & W \succcurlyeq \diag{\Delta} \succcurlyeq 0\,,\nonumber\\
    & \mathbf{1}^\top \Delta \leq b\,.\nonumber
\end{align}
It is easy to observe that the objective is linear and the constraints are convex in $\Delta$, $u$, and $t$. Therefore, (\ref{polySovable.eqn}) is an SDP which can be solved in polynomial time.
\end{proof}
\section{Numerical Simulations}
\label{exper.sec}

In this section, we examine the effectiveness of the proposed curing policies as well as the network design algorithms through numerical simulations. The model for the SIS process is described in Section~\ref{model.subsec}. For all numerical examples, the original infection rates of the edges in the network are sampled from a uniform distribution of the range $[0.4, 1.6]$. Curing rates of the nodes are allocated by the policies with a total budget $r$.

\subsection{Simulations for Curing Policies}
We compare Algorithm~\ref{alg:apprImpe} against 4 baseline curing policies, including two static policies and their dynamic variations.
\begin{enumerate}
    \item Uniform (static): the curing budget is uniformly allocated to all nodes, $\rho_u(t)=1/n, \forall u\in V$.
    \item Degree (static)~\citep{BCGS10}: the curing rate of a node $u\in V$ is set to $r\cdot d_u/(\sum_{u\in V} d_u), \forall u\in V$.
    \item Uniform (dynamic): the curing budget is uniformly allocated to all \emph{currently infected} nodes, i.e.\ $\rho_u(t) = 1/\sizeof{I(t)}, \forall u\in I(t)$, and otherwise set to $0$.
    \item Degree (dynamic): the curing rate of a node $u\in I(t)$ is set to $\rho_u(t) = r\cdot d_u I_u(t)/(\sum_{u\in V} d_u I_u(t))$, and otherwise set to $0$.
\end{enumerate}

We implement the CURE policy by using Algorithm~\ref{alg:apprImpe} to calculate the ordering of curing. We note that instead of using the $1/3$-balanced cut algorithm of~\cite{LR99}, we use a heuristic for balanced cut based on the eigenvector corresponding to the second smallest eigenvalue of the graph Laplacian. We sort the nodes using their corresponding values in the eigenvector, and run a sweep algorithm to find the best $1/3$-balanced cut, similar to the sweep algorithm for the sparsest cut problem~\citep{Chu97}. Our study on the fair curing policy shows the existence of polynomial time approximations to the fair curing problems. However, the construction of the tree embedding~\citep{Rac08} remains a theoretical tool, not a practical algorithm. Therefore, we do not implement Algorithm~\ref{alg:fairapprImpe}.

To test the effectiveness of the CURE policy using Algorithm~\ref{alg:fairapprImpe}, we simulate the SIS process on two typical networks: the locally connected network and the binary tree network. The locally connected network is augmented from a path of $n=3000$ vertices. Suppose the nodes in a path are labeled using its distances to one of the endpoint of the path, then we add edges $\{i, i+2\},  i =0,1,\ldots,(n-2)$ to the path to get the locally connected network with $m=5997$ edges. The binary tree network is a perfect binary tree with $11$ layers, therefore it has $n=2047$ nodes and $m=2046$ edges.


\begin{figure}[htbp]
    \centering
    \includegraphics[width=0.4\linewidth]{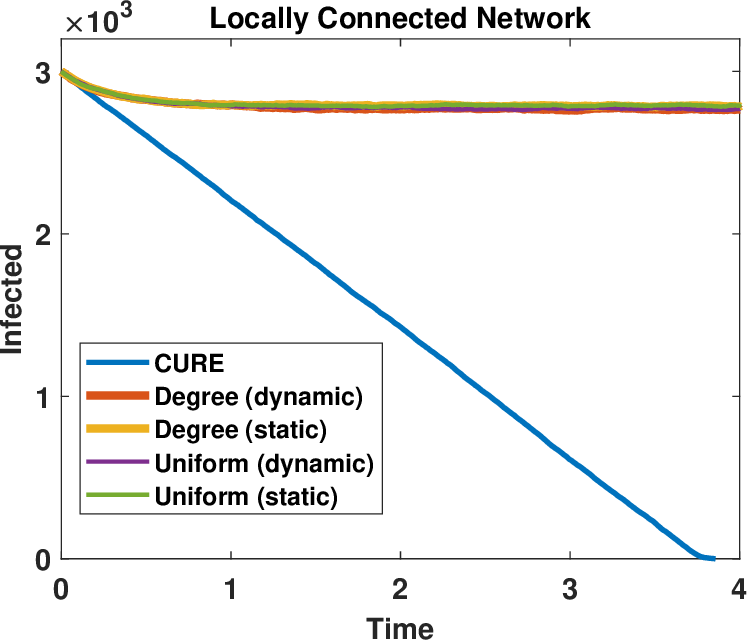}\quad
    \includegraphics[width=0.4\linewidth]{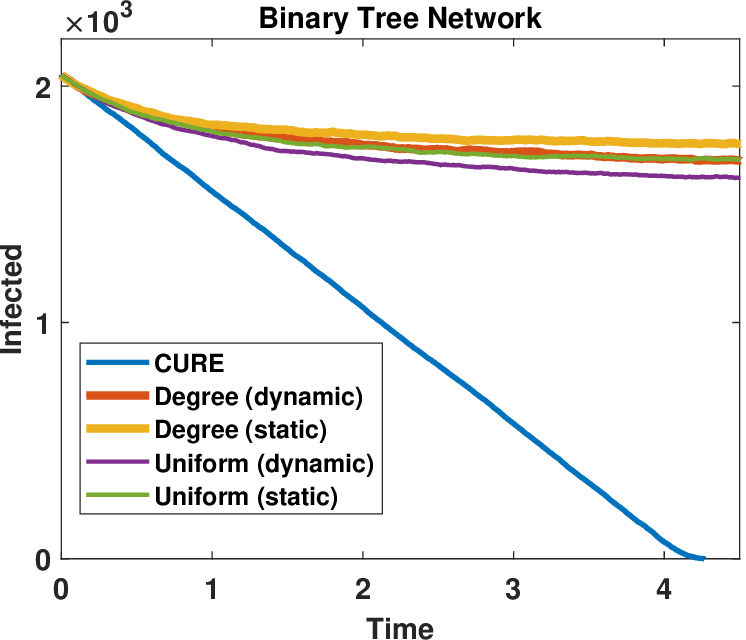}
    \caption{Comparison of CURE (using Algorithm~\ref{alg:apprImpe}) with baseline policies}
    \label{fig:curingPolicies}
\end{figure}

Figure~\ref{fig:curingPolicies} shows the performance of the CURE policy and the baselines. The figure shows the average results across $10$ runs for all methods. The curing budgets are set to $800$ and $500$ for the locally connected network and binary tree network, respectively. For the same network, all policies use the same curing budget. We assume all nodes in the networks are infected in the beginning of the process.  

In both networks, the CURE policy is the only one that succeeds in curing all nodes in all $10$ runs. All baseline methods fail to cure the network in any run. The figure clearly shows the advantage of the CURE policy against the considered heuristics, whether they are static or dynamic. 

\subsection{Simulations for Network Design Algorithms}
For the network design algorithms, we simulate the process on a human contact network and an email network. We use the networks to simulate the spread of human viruses and internet worms respectively. The human contact network is constructed by collecting data using mobile technologies~\cite{klepac2018contagion,firth2020using}. The network has $405$ nodes and $1258$ edges. We let the number of initial infections to be $200$, and the curing budget be $500$. All initial infected nodes are selected uniformly at random. The email network was generated using email data from a research institute~\cite{leskovec2007graph}. We ignore the directions of edges, and then attain a network with $986$ nodes and $16057$ edges. We let the number of initial infected nodes to be $200$, and $r=10000$ for the email network.

For both networks, We compare CURE policy with CURE policy augmented by linear programming (LP)-based network design, as well as random curing budget allocation complemented by semi-definite programming (SDP)-based network design. Figure~\ref{fig:networkDesign} shows the average trajectories of $10$ runs for each method for both networks. Simulations on both networks show consistent results. 

For the human contact network, the trajectory of the CURE policy clearly shows a waiting period at the beginning of the process. After the curing rates are dynamically allocated, the number of infections plateaus subsequent to a transient declination. With network design algorithms, the curing processes are continues and quickly stop the SIS processes. We note that the SDP approach considers the worst curing rate allocation strategy and hence is more conservative. The SDP network algorithm achieves shorter extinction time with random curing targets at any given time, with an average of $890.24$ network weight reduction for the human contact network, while the LP algorithm only reduces $419.11$ edge weights in average. Results are similar for the email network, where SDP and LP costs $4994.05$ and $2595.09$ weight reduction in average, respectively. The averages are taken over all network design periods until the end of the processes.

\begin{figure}[htbp]
    \centering
    \includegraphics[width=0.4\linewidth]{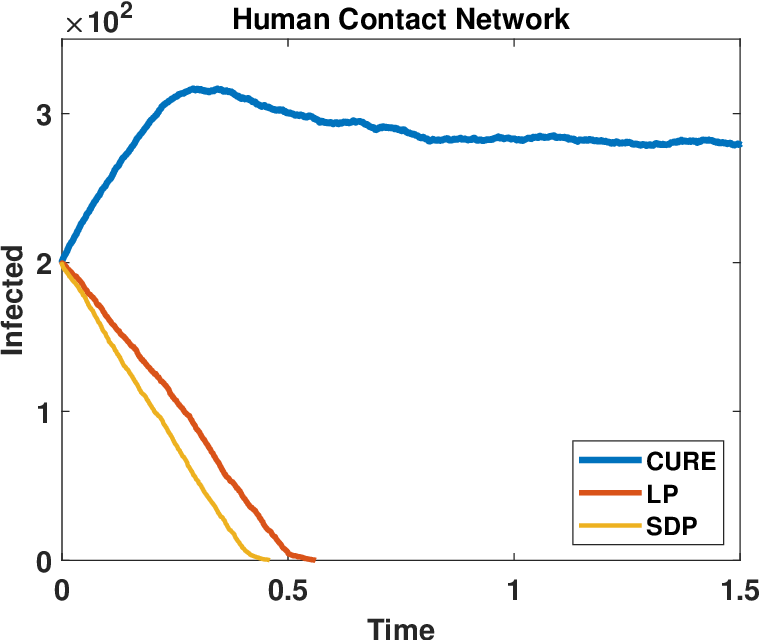}\quad
    \includegraphics[width=0.4\linewidth]{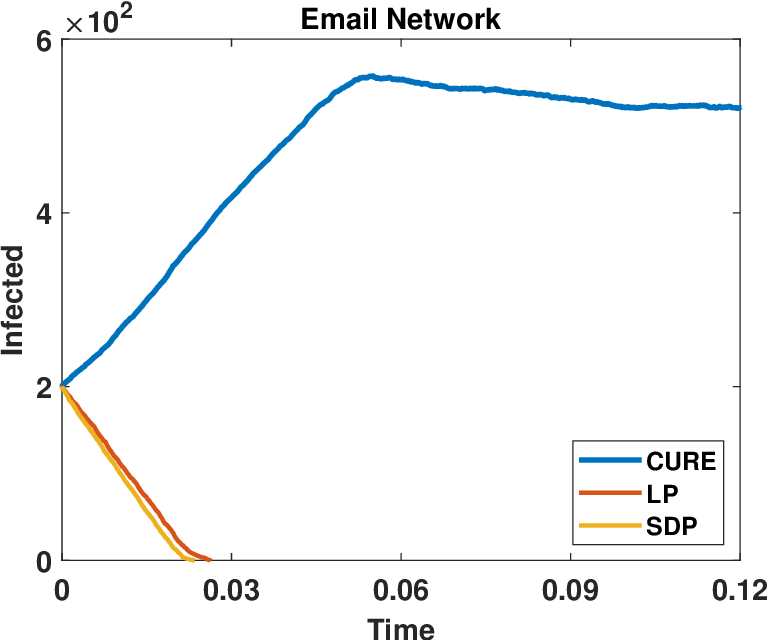}
    \caption{Comparison of network design algorithms: 1) CURE without network design; 2) CURE with linear programming network design; 3) Random curing rate allocation with SDP.}
    \label{fig:networkDesign}
\end{figure}

\section{Conclusion}
\label{concl.sec}
In this paper, we have studied efficient algorithms for dynamic curing policies of SIS epidemic models. We have proposed approximation algorithms to the impedance and a fair impedance of an infected set in a given graph. We have also proposed algorithms for network design problems to help cure the network. Provable guarantees have been provided for all the studied algorithms. We have shown the effectiveness of the policy and network co-design approach by running simulations on real contact networks.

There are two open problems for consideration in future research. The first open problem is to find algorithms ensuring fairness in \emph{every} interval longer than a given constant. The second is to devise algorithms for optimal network design to minimize impedance \emph{without} fixing the order of curing. Future work should also study practical algorithms for fair curing, and improving the running time of current solutions.

\section*{Acknowledgement}
The first author is supported in part by the National Natural Science Fundation of China under grant 62303338 and in part by the Fundamental Research Funds for the Central Universities, under grant Sichuan University YJ202285. The second author is supported in part by the National Science Foundation  under grants CCF-1955351 and CCF-1934931. The third  author  is supported in part by the Microsoft  Corporation  and  in  part  by  the  National  Science  Foundation  under  grants  NFS-CNS-2028738 and NFS-ECCS-2032258. The fourth author is supported in part by Knut and Alice Wallenberg Foundation and by the Swedish Research Council.

\appendix

\section{Unweighted Edge Deletion}
\label{ued.appdx}
We show a polynomial time reduction from UWCMP to the following problem.
\begin{problem}[\underline{I}nterval \underline{S}cheduling \underline{P}roblem on $\underline{\ell}$ machines ($\ell$-ISP)]
Given a set of $m$ intervals $\Pi = \{\pi_1,\dots, \pi_m\}$ where each interval $\pi_i\in \Pi$ has a start time $s_i$ and a finish time $f_i$ ($s_i<f_i$), find a maximal subset of intervals $Q\subseteq \Pi$ and a mapping $\sigma:Q\mapsto [\ell]$ such that there is no overlap among intervals scheduled on every machine. 
\end{problem}
\begin{theorem}
There exists a polynomial time algorithm which finds the optimal solution for UWCMP.
\end{theorem}
\begin{proof}
For an instance of UWCMP, nodes are ordered according to the crusade $p$. In particular, the node in $p_i 
\setminus p_{i-1}$ is assigned with the label $v_i$ for $i\in[k]$. 
Then the ordered vertices $\{v_i\}_{i=1}^k$ in UWCMP are mapped to timestamps in $\ell$-ISP. The edges are mapped to intervals. Specifically, for an edge $e=(v_i,v_j)$, we let $s_e=\min\{v_i,v_j\}$ and $f_e=\max\{v_i,v_j\}$. We further map the threshold to the number of machines by defining $\ell=b$. Then we obtain the input for $\ell$-ISP from an instance of UWCMP. We note that $\ell$-ISP can be solved by a greedy algorithm in $O(m \log m)$ running time~\cite{FN95,CL95}.
From the output of an instance of $\ell$-ISP, we can construct a feasible solution for the corresponding UWCMP instance. For the interval set $Q$, we find the corresponding edges $E_Q$ in UWCMP and let $\Delta_{uv}=0$ for $(u,v)\in E_Q$. We further let $\Delta_{uv}=1$ for $(u,v)\in E 
\setminus
E_Q$. The cost of UWCMP is given by $|E|-|Q|$. The running time of the algorithm for UWCMP is $O(m\log m)$ where $m$ is the number of edges of $G$ in UWCMP.
\end{proof}

\section{Analysis for the Fair Curing Algorithm}
\label{appdx:fair}

In Algorithm ~\ref{alg:fairapprImpe}, We sequentially find fair partitions for the given checkpoints. For each checkpoint, we find a fair partition at this checkpoint.
Then, we run balanced cuts recursively to pick the crusade for nodes before this checkpoint. Then we run the algorithm on the subgraph supported on the remaining nodes.

We reduce the problem to finding a fair partition with a small cut on tree instances. To this end, we recall the definition of hierarchically well-separated trees.

\begin{lemma}[\cite{Rac08}]
For any graph $G=(V,E,w)$ with a weight function $w: E\mapsto \mathbb{R}_{\geq 0}$, there exists a collection of trees  $T_1=(V,E_1), T_2=(V, E_2), \ldots, T_N=(V, E_N)$ with tree $T_i$ having an edge weight function $w^i: E_i\mapsto \mathbb{R}_{\geq 0}$, and find non-negative multipliers $(\lambda_1,\ldots,\lambda_N)$, such that $\sum_{i=1}^N\lambda_i=1$ and $N=\mathrm{poly}(|V|)$. For any subset $S\subseteq V$ let $\calB(S)$ be the set of edges in $E$ with one endpoint in $S$, and $\calB_i(S)$ be the set of edges in $E_i$ with one endpoint in $S$. Then, for any $S\subseteq V$:
\begin{enumerate}
    \item $w(\calB(S))\leq w^i(\calB_i(S))$ for every $i\in [N]$;
    \item $\sum_{i=1}^N \lambda_i w_{i}(\calB_i(S))\leq O(\log n)w(\calB(S))$.
\end{enumerate}
\end{lemma}

\begin{proof}[Proof of Theorem~\ref{thm:fair}]
We prove a simple setting with two groups $V_1, V_2$ and one checkpoint $\tau_1$. We first partition the bag $A$ into $\gamma$-fair parts $S_0$ and $S_1$ by using the hierarchical well-separated trees and dynamic programming on trees. Then we partition $S_0$ and $S_1$ with balanced cuts recursively until the size of a set is $1$.

Let $S_0^*, S_1^*$ be the minimum weight $\gamma$-fair partitioning of bag $A$. By the hierarchical well-separated trees, there exists a tree embedding $T^*$ such that the total weight of edges crossing the sets $S_0^*, S_1^*$ is preserved up to a $O(\log k)$ factor. Suppose we find the minimum weight $\gamma$-fair partitioning on the tree $T^*$, which is shown in Lemma~\ref{lem:dp}. Then, we find a $\gamma$-fair partitioning which has a total cut weight at most $O(\log k)$ times the optimal $\gamma$-fair impedance.

Finally, we pick the order for two parts $S_0$ and $S_1$ separately with recursive balanced cuts. Similar to Theorem~\ref{thm:impedance}, there are totally $O(\log k)$ balanced cuts and each cut has a total weight at most $O(\log k)$ times the optimal $\gamma$-fair impedance. Combining all these balanced cuts and the cut edges crossing $S_0$ and $S_1$, we get the conclusion.
\end{proof}

\begin{lemma}\label{lem:dp}
Given a tree $T$, there exists a polynomial time dynamic programming algorithm that finds the minimum cut that partitions the tree $T$ into $\gamma$-fair parts.
\end{lemma}
\begin{proof}
As explained in Lemma 15.18 in~\cite{WS11}, we can assume without loss of generality that the tree is binary.
Our dynamic programming algorithm
is based on the value table $$
DP[v, l_v, k_1(0), k_2(0)],
$$
where $k_h(0)$, $h\in [\ell]$ denotes the number of nodes in the sub-tree rooted at $v$ and group $h$ which are assigned to part $S_0$.
The indicator variable $l_v$ takes value $0$ if $v \in S_0$, otherwise $l_v = 1$. We use $DP[v, l_v, k_1(0), k_2(0)]$ to record the minimum cut in the subtree rooted at $v$ that satisfies the given configuration.

Then we calculate the value table from the leaves of the tree to the root.
We set $DP[v, s_v, k_1(0), k_2(0)]$ to the minimum of the following four cases:

Case 1: When $l_v=l_{v^L}=l_{v^R}$, we obtain
\begin{align}\label{eq:dp}
    &\min\{DP[v^L,l_{v^L},  k^L_1(0), k^L_2(0)] + DP[v^R,l_{v^R},  k^R_1(0), k^R_2(0)],  \\
   &\text{where\,} k_h(0) = k^L_h(0) + k^R_h(0) + \1\{v \in V_h, l_v = 0\}, \forall h\in \{1,2\}\}. \nonumber
\end{align}
To simplify the notation, we define \eqref{eq:dp}
to be
$f(l_v,l_{v^L},l_{v^R})$.

Case 2: When $l_v=l^L_v\neq l^R_v$, we have $w(v, v^R)+ f(l_v,l_{v^L},l_{v^R})$.

Case 3: When $l_v=l^R_v\neq l^L_v$, we have $w(v, v^L) + f(l_v,l_{v^L},l_{v^R})$.

Case 4: When $l_v\neq l^R_v= l^L_v$, we have $w(v, v^L)+ w(v, v^R) + f(l_v,l_{v^L},l_{v^R})$.

The first case corresponds to cutting neither of $\{v, v^L\}$ and $\{v, v^R\}$; the second to cutting $\{v, v^L\}$ but not $\{v, v^R\}$; the third to cutting $\{v, v^R\}$ but not $\{v, v^R\}$; the last to cutting both edges.

To fill in the table, we initialize all leaves of the tree. Let $v$ be a leaf, then $DP[v, l_v, \1\{v \in V_1, l_v = 0\}, \1\{v \in V_2, l_v = 0\}]=0$ for $l_v\in\{0,1\}$. Let all other entries of leaves be infinity. We make a pass from the bottom to fill in the table. Then we check the root of the tree and find the minimum value in the table such that the partitioning is $\gamma$-fair. The running time of the algorithm is $O(k^3)$.
\end{proof}

Next, we show that our algorithm achieves a $O(\log^2 k)$ approximation for a set of checkpoints $\calT$
when the
set of checkpoints satisfies a doubling condition.

\begin{proof}[Proof of Theorem~\ref{thm:fair-multiple-checkpoint}]
    Similar to Theorem~\ref{thm:fair}, it is sufficient to show that we partition the bag $A$ into $\gamma$-fair parts w.r.t. the set of checkpoints $\calT$. In Theorem~\ref{thm:fair}, we show that our algorithm finds a $\gamma$-fair partitioning w.r.t. one checkpoint.  To find this $s+1$ partitioning, we iteratively call our algorithm at each checkpoint from $\tau_1$ to $\tau_s$. Specifically, we first run our algorithm to find a $\gamma$-fair partitioning w.r.t. $\tau_1$, $S_1$ and $A\backslash S_1$. Then, we partition the set $A\backslash S_1$ w.r.t. checkpoint $\tau_{2}$. We repeat this process until we get a $s+1$ partitioning of bag $A$.

    For every $i \in \{1,2,\cdots,s\}$, we show that the total weight of cut edges crossing the checkpoint $\tau_i$ in the algorithm solution, $\cup_{j=0}^{i-1} S_j$ and $\cup_{j=i}^s S_j$ is at most $O(\log^2 k)$ times the optimal $\gamma$-fair impedance. By Theorem~\ref{thm:fair}, it holds for the checkpoint $\tau_1$. Next, we consider any checkpoint $\tau_i > \tau_1$. We show that there exists a $\gamma$-fair partitioning of nodes in $\cup_{i=0}^i S_i$ at checkpoint $\tau_i$ such that the weight of cut edges is at most two times the optimal $\gamma$-fair impedance. By Theorem~\ref{thm:fair}, our algorithm finds a $O(\log k)$ approximation to this partitioning. Since there are at most $\log_2 k$ such checkpoints due to the doubling condition, the total weight of cut edges in our solution is at most $O(\log^2 k)$ times the optimal $\gamma$-fair impedance.

    Now we show the existence of a good partitioning in the set of nodes $S' = \cup_{j=i-1}^s S_j$ for every $i > 1$. Let $p^*$ be the optimal $\gamma$-fair crusade w.r.t. the set of checkpoints $\calT$. We denote the nodes contained in the optimal crusade $p^*$ before checkpoint $\tau_{i-1}$ but not contained in our solution before $\tau_{i-1}$ by
    $$S'_1 = (p^*_0 \backslash p^*_{\tau_{i-1}})\backslash \cup_{j=0}^{i-2} S_j.$$
    We have $S'_1$ is a subset of $S'$ and $\sizeof{S'_1} \leq \tau_{i-1}$. Since checkpoints satisfy $\tau_{i} - \tau_{i-1} \geq \tau_{i-1}$, we fully assign this set $S'_1$ to the part between the checkpoints $\tau_{i-1}$ and $\tau_{i}$.
    Note that the set $S' \backslash S'_1$ is contained in $p^*_{\tau_{i-1}}$ and at least $\tau_{i}-\tau_{i-1} - \sizeof{S'_1}$ among them are between checkpoints $\tau_i$ and $\tau_{i-1}$.
    Thus, let $S'_2$ be the set of the first $\tau_{i}-\tau_{i-1} - \sizeof{S'_1}$ nodes in $p^*_{\tau_{i-1}}$ that are also contained in $S'
    \setminus
    S'_1$. Then $S''= S'_1 \cup S'_2$ is a valid partition of length $\tau_{i}-\tau_{i-1}$.  Figure~\ref{fig:crusades} illustrates the
    construction of the partitioning.

    Then, we show that the cut $S''= S'_1 \cup S'_2$ is a good partitioning in $S'$. The total weight of this cut $S''$ in $S'$ is at most the sum of the weight of edges from $S'_1$ to $(S'')^c$ and edges from $S'_2$ to $(S'')^c$, which is at most two times the width of $p^*$. Note that $S'_1$ is a subset of $p^*_0      \setminus     
    p^*_{\tau_{i-1}}$, which is a $\gamma$-fair set with size $\tau_{i-1} \leq \tau_{i}-\tau_{i-1}$. The set $S'_2$ is a subset of $p^*_{\tau_{i-1}}
    \setminus
    p^*_{\tau_{i}}$, which is also a $\gamma$-fair set with size $\tau_{i}-\tau_{i-1}$. Therefore, we have $S''$ is a $2\gamma$-fair set.
\end{proof}

\begin{figure}[htbp]
    \centering
    \includegraphics[width=.5\linewidth]{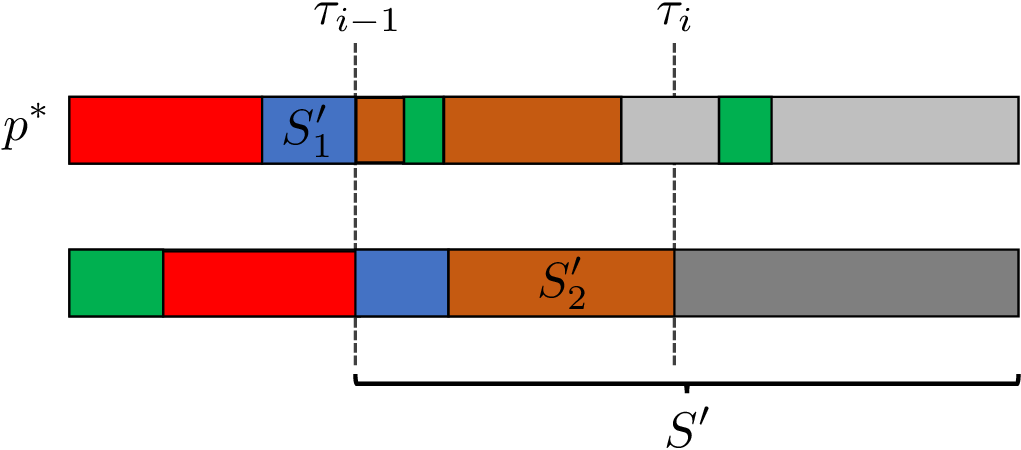}
    \caption{The construction of a good partition in $S'$. The top figure corresponds to the optimal $\gamma$-fair crusade. The bottom figure corresponds to a good partition.}
    \label{fig:crusades}
\end{figure}

\end{document}